\documentclass[JMC]{degruyter-journal}

\usepackage{boxedminipage}
\newtheorem{claim}{Claim}[section]

\volume{5}
\pubyear{2011}
\doi{XXX}
\communicated{Spyros Magliveras}
\received{4 May, 2011}
\revised{16 January, 2012}

\title{On  Bringer-Chabanne EPIR Protocol for Polynomial Evaluation}

\lastnameone{Chee}
\firstnameone{Yeow Meng}
\emailone{ymchee@ntu.edu.sg}
\addressone{\\Division of Mathematical Sciences, \\School of Physical and Mathematical Sciences, \\Nanyang Technological University, \\
21 Nanyang Link}
\countryone{ \\Singapore 637371}

\lastnametwo{Wang}
\firstnametwo{Huaxiong}
\emailtwo{hxwang@ntu.edu.sg}
\addresstwo{\\Division of Mathematical Sciences, \\School of Physical and Mathematical Sciences, \\Nanyang Technological University, \\
21 Nanyang Link}
\countrytwo{ \\Singapore 637371}

\lastnamethree{Zhang}
\firstnamethree{Liang Feng}
\emailthree{liangfeng.zhang@ntu.edu.sg}
\addressthree{\\Division of Mathematical Sciences, \\School of Physical and Mathematical Sciences, \\Nanyang Technological University, \\
21 Nanyang Link}
\countrythree{\\Singapore 637371}

\acknowledgments{The research is supported in part by the Singapore National
Research Foundation under Research Grant NRF-CRP2-2007-03.}

\abstract{
Extended private information retrieval (EPIR) was defined by
\cite{BCPT07}
at CANS'07 and generalized by
\cite{BC09}
at
AFRICACRYPT'09. In the generalized setting, EPIR allows a user to
evaluate a function on a database block such that the database
can learn neither which function has been evaluated nor on which
block the function has been evaluated and the user learns no more
information on the database blocks except for the expected result. An
EPIR protocol for evaluating  polynomials over a finite field $L$
was proposed by Bringer and Chabanne
 in \cite{BC09}.
We show that the protocol does not
satisfy the correctness requirement as they have claimed. In
particular, we show that it does not give the user the
expected result with large probability if one of the coefficients of
the polynomial to be evaluated is  primitive in $L$ and the others
belong to   the prime subfield of $L$.
}



\begin{document}

\section{Introduction}

Extended private information retrieval (EPIR) was  motivated by
privacy-preserving biometric authentication and formally defined in
\cite{BCPT07}. It enables a user to privately evaluate a fixed and
public function with two inputs, one chosen block from a database
and one additional string. Two EPIR protocols were proposed in
\cite{BCPT07}. One is
 for testing
equality and the other is for computing weighted Hamming distance.
As a  cryptographic primitive, EPIR has been generalized by
\cite{BC09} in order to attain more flexibility. In the generalized
setting, the function to be evaluated is neither fixed nor public.
Instead, it is chosen from a set of public functions by the user. A
new EPIR protocol
 in the generalized setting was  proposed in \cite{BC09}.
As noted in \cite{BCPT07}, EPIR is indeed a combination of private
informatrion retrieval \cite{CGKS95} and general secure two-party
computation \cite{Gol04}.

{\bf Related Work.} Private information retrieval (PIR) was
introduced by \cite{CGKS95}. It allows a user to retrieve a data
item from a  database  such that
 the database cannot learn which item the user is interested in.
The requirement on the privacy of the identity of the retrieved data
item is  called {\em user privacy}. The main measure of the
efficiency of a PIR protocol is its  {\em communication complexity},
i.e., the total number of bits exchanged by the user and the
database for retrieving a {\em single bit}.
 PIR protocols have been constructed in both the {\em information-theoretic} setting
  \cite{CGKS95,Amb97,BIKR02,BIK05,WY05,Yek07,Efr09,IS10,CFLWZ10} and the {\em computational}
  setting \cite{CG97,KO97,CMS99,KO00,YS01,Chang04,GR05,Lip05,GKL10}.
In an information-theoretic PIR protocol, the database learns {\em
absolutely}  no information on which item the user is interested in
even if it has unlimited computing power. On the other hand, in a
computational PIR (CPIR) protocol, the identity of the retrieved
data item is not revealed only if the database is polynomial-time
and cannot efficiently solve certain number-theoretic problems,
i.e., certain cryptographic assumptions hold. For example, the PIR
protocol of \cite{CG97} is a {\em two-database} CPIR protocol in
which each database cannot figure out which item the user is
interested in under the assumption that one way functions exist.
EPIR protocols of \cite{BCPT07,BC09} are mostly close to the {\em
single-database} CPIR protocols. The first single-database CPIR
protocol was proposed by \cite{KO97}. It achieves the user privacy
under the assumption that deciding quadratic residuosity is hard and
has communication complexity $O(N^c)$ for any small constant $c>0$,
where $N$ is the size of the database. Subsequently,  \cite{CMS99}
constructed  a single-database CPIR protocol of communication
complexity $O(\log^8(N))$  under the  $\Phi$-hiding assumption.  So
far, the most efficient single-database CPIR protocol was obtained
by \cite{GR05} under the assumption that the decision subgroup problem
is hard.  It requires the user to exchange $O(k+d)$ bits with the
database for retrieving $d$ bits, where $k\geq \log N$ is the
security parameter.
Other constructions of single-database CPIR  protocols can
be found in \cite{KO00,YS01,Chang04,GKL10}.

PIR does not provide any privacy for the database. Typically, the
user may obtain  a large number of data items in an execution of a
PIR protocol. In order to prevent the user from obtaining more than
one data item in any execution of a PIR protocol, \cite{GIKM98}
introduced the notion of {\em data privacy} and proposed
transformations from information-theoretic PIR protocols to  the
so-called symmetrically private information retrieval (SPIR)
 protocols
 which meet the data privacy. The SPIR protocols of \cite{GIKM98} are in the information-theoretic setting.
SPIR can be defined  in the computational setting as well. Following
the security definition of general secure two-party and multi-party
computation \cite{Gol04}, in the computational setting, a PIR
protocol is said to achieve  data privacy if, for any query, the
user cannot tell whether it is interacting with a real-database
which has $N$ data items or a simulator which only knows the
retrieved data item.
 Interestingly, single-database SPIR protocols in  the computational setting  are essentially
 communication-efficient 1-out-of-$N$ Oblivious transfer (OT)
 \cite{Rab81,EGL85,BCR87,GMW87,Kil88} protocols.
Oblivious transfer \cite{Rab81} is a fundamental  cryptographic
primitive, on which any secure two-party and multi-party computation
can be built  \cite{Kil88} in an unconditionally secure way. A
1-out-of-$N$ OT allows a receiver Bob to choose one of the $N$
secrets held by a sender Alice such that Alice learns no information
on Bob's choice and Bob cannot learn more except the secret he
chooses. \cite{NP99} proposed transformations from any PIR protocols
to SPIR protocols in the computational setting. Their transformation
requires only one execution of a given PIR protocol and $\log N$
executions of a 1-out-of-2 OT protocol. The notion of EPIR
\cite{BCPT07,BC09} is essentially a generalization of SPIR in the
computational setting.

EPIR is also  related to  selective private function evaluation
\cite{CIKRRW01}, oblivious polynomial evaluation  \cite{NP99} and
private keyword search \cite{FIPR05}. A selective private function
evaluation protocol \cite{CIKRRW01} allows a client to privately
evaluate a public function on the inputs  held by one or more
servers. Comparing with EPIR, the client  only decides on which
inputs the public function will be evaluated. An oblivious
polynomial evaluation protocol \cite{NP99} allows a receiver to
privately evaluate  a polynomial function on  his  input, where the
polynomial is held by a sender. Comparing with EPIR, the function to
be evaluated is not known to the receiver and the input on which the
function is evaluated is not known to the sender. A private keyword
search protocol  \cite{FIPR05} allows a client to privately search a
database with a keyword such that he learns the associated record if
the keyword is contained in the database and learns nothing
otherwise. In a sense, EPIR can  also be seen as a generalization of
the above problems.

{\bf Results.} The protocol described in Section 4.3 of \cite{BC09}
will be our main topic in this paper and termed as Bringer-Chabanne
EPIR protocol   from now on. It was claimed \cite{BC09} that the
protocol enables  a user to privately evaluate any polynomial
$F(t)\in L[t]$ on a chosen  database block $R_i$, where $L={\rm
GF}(p^n)$ is the field extension of degree $n$ of the prime field
$K={\rm GF}(p)$. We study the correctness of the Bringer-Chabanne
EPIR protocol   and show that it may fail frequently. In particular,
we show that, by executing the protocol,  the user with input
$(F(t),i)\in L[t]\times [N]$
 does not
learn the expected result (i.e., $F(R_i)$) with a large probability
if $F(t)\in \mathcal{P}$, where $
\mathcal{P}=\{f(t)=\sum_{k=0}^{d}f_kt^k: \exists~ 0\leq l\leq
d{\rm ~such~that~}f_{l}\in L{\rm~is~of~order~}
p^n-1{\rm~and~}f_k\in K{\rm ~for~every~}k\neq l\} $.

{\bf Methodology.} Our argument is by contradiction. To simplify the
argument, we first give a restricted version of the Bringer-Chabanne
EPIR protocol.  In the restricted version, the database is
deterministic and only has one block, i.e., $N=1$. We
note that if the Bringer-Chabanne EPIR protocol satisfies the
correctness requirement, then so does the restricted version. We
then show that the restricted version does not satisfy the
correctness requirement if the polynomial to be evaluated is in
$\mathcal{P}$. This result allows us to conclude that the
Bringer-Chabanne EPIR protocol   does not satisfy the correctness
requirement as \cite{BC09} has claimed.

{\bf Organization.} The remainder of this paper  is organized as
follows. In Section 2, we recall the definition and security model
of EPIR \cite{BC09}. In Section 3, we recall the Bringer-Chabanne
EPIR protocol. In Section 4, we give a restricted version of the
Bringer-Chabanne EPIR protocol   and   show that the restricted
version fails frequently if the polynomial to be evaluated is in
$\mathcal{P}$. At last, in Section 5, we conclude the paper.

\section{Preliminaries}

\subsection{Definition}

Following the definition of \cite{BC09}, a single-database EPIR
protocol is a protocol between a database $\mathcal{DB}$ who has $N$
blocks   $(R_1,\ldots, R_N)\in (\{0,1\}^{l_1})^{N}$ and a user
 $\mathcal{U}$ who wants to  evaluate  $F(R_i)$ for a function
$F\in \mathcal{F}$ and an index $i\in[N]$,  where $\mathcal{F}$ is a
 set of functions from $\{0,1\}^{l_1}$ to $\{0,1\}^{*}$ and public. Such a protocol allows
 $\mathcal{U}$ to learn $F(R_i)$ but no more information on the database blocks while $\mathcal{DB}$
learns no information on $(F,i)$.

The above definition of EPIR is a generalization of  \cite{BCPT07}
and provides the user with more flexibility of choosing the function
$F$ from a large set  $\mathcal{F}$.
 In the context of this definition, the EPIR
for testing equality  \cite{BCPT07} has  $\mathcal{F}=\{{\rm IsEqual}(\cdot, X):  X\in
\{0,1\}^{l_1}\}$, where ${\rm IsEqual}(R_i,X)=1$ if $R_i=X$ and 0
otherwise.
The
EPIR  for computing weighted Hamming distance \cite{BCPT07}
has $\mathcal{F}=\{d_w(\cdot, X):  X\in
\{0,1\}^{l_1}, w\in \mathbb{N}^{l_1}\}$, where
$d_w(R_i,X)=\sum_{j=1}^{l_1} w_j\cdot (R_i^{(j)}\oplus X^{(j)})$
 (For every $j\in[l_1]$, $R_i^{(j)}$ and $X^{(j)}$ are the $j$-th bits of $R_i$ and
$X$, respectively).

\subsection{Security Model}\label{sec:security_model}

As in \cite{BC09,BCPT07}, we denote by ${\bf retrieve}(F,i)$ the
query made by a user with input $(F,i)\in \mathcal{F}\times [N]$.
Without further notice, algorithms are assumed to be
polynomial-time. If an algorithm $\mathcal{A}$ runs in $k$
stages,  then we shall write $\mathcal{A} = (\mathcal{A}_1,\mathcal{A}_2,
\ldots, \mathcal{A}_k )$. The security is evaluated by an experiment between an
attacker and a challenger, where the challenger simulates the
protocol executions and answers the attacker's oracle queries.
 For $\mathcal{A}$ a probabilistic algorithm, we denote
by $\mathcal{A}(\mathcal{O}, {\bf retrieve})$
 the action to run $\mathcal{A}$ with access to any polynomial number of ${\bf retrieve}$
 queries
generated or answered (depending on the position of the attacker) by
the oracle $\mathcal{O}$. A function $\tau: \mathbb{Z}\rightarrow
\mathbb{R}$ is said to be {\em negilible} if for any polynomial $P$,
there is an integer $N_P$ such that $\tau(n)\leq 1/P(n)$ for every
$n\geq N_P$. If $\tau(n)$ is  negilible, then $1-\tau(n)$ is said to
be
{\em overwhelming}.\\
{\bf Correctness.} An EPIR protocol is said to be {\em correct} if
any query ${\bf retrieve}(F,i)$
 returns the correct value of $F(R_i)$ with an overwhelming probability when $\mathcal{U}$ and
 $\mathcal{DB}$ follow the protocol specification.\\
{\bf User Privacy.} Informally, an EPIR protocol is said to {\em
respect  user privacy} if for any query ${\bf retrieve}(F, i)$,
$\mathcal{DB}$ learns no information on  $(F,i)$. Formally, an EPIR
protocol is said to respect user privacy if any attacker
$\mathcal{A} = (\mathcal{A}_1, \mathcal{A}_2, \mathcal{A}_3,
\mathcal{A}_4)$, acting as a malicious database, has only a
negligible advantage $|\Pr[b^{\prime}=b]-\frac{1}{2}|$ in the
following experiment:

$\hspace{0.2cm}{\bf Exp}_{\mathcal{A}}^{\rm user\text{-}privacy} $
\vspace{-0.25cm}
$$
 \left| \begin{array}{ccl}
               (R_1,\ldots,R_N) &\leftarrow&\mathcal{A}_1(1^l)    \\
              1\leq i_0,i_1\leq N; F_0,F_1\in \mathcal{F}&\leftarrow& \mathcal{A}_2(Challenger;{\bf retrieve})      \\
              b&\leftarrow& \{0,1\}     \\
              \emptyset &\leftarrow& \mathcal{A}_3(Challenger; {\bf retrieve}(F_b,i_b))      \\
             b^{\prime} &\leftarrow& \mathcal{A}_4(Challenger;{\bf retrieve})      \\
\end{array}
\right.$$ {\bf Database Privacy.} Informally, an EPIR protocol is
said to {\em respect database privacy}
 if a malicious user $\mathcal{U}$ cannot learn more information than $F^{\prime}(R_{i^{\prime}})$ for
some $(F^{\prime}, i^{\prime})\in \mathcal{F} \times [N]$  via a
query ${\bf retrieve}$. This intuitive description  can be
formalized via simulation principle by saying that the user
$\mathcal{U}$ cannot determine whether he is interacting with a
simulator which takes  only $(i^{\prime},
F^{\prime}(R_{i^{\prime}}))$ as input, or with $\mathcal{DB}$. We
denote by $\mathcal{S}_0$ the database $\mathcal{DB}$. Formally, an
EPIR protocol is said to respect database privacy if there is a
simulator
 $\mathcal{S}_1$,  which receives  an
auxiliary input $(i^{\prime}, F^{\prime}(R_{i^{\prime}}))$ from a
{\em hypothetical oracle}
 $\mathcal{O}$ for every query
${\bf retrieve}$, such that any attacker $\mathcal{A} =
(\mathcal{A}_1,\mathcal{A}_2)$, acting as a malicious user, has only
a negligible advantage $|\Pr[b^{\prime}=b]-\frac{1}{2}|$ in the
following experiment:

$\hspace{2.6cm}{\bf Exp}_{\mathcal{A}}^{\rm database\text{-}privacy}
$ \vspace{-0.25cm}
$$
 \left| \begin{array}{ccl}
              b&\leftarrow& \{0,1\}     \\
               (R_1,\ldots,R_N) &\leftarrow&\mathcal{A}_1(1^l)    \\
             b^{\prime} &\leftarrow& \mathcal{A}_2(\mathcal{S}_b;{\bf retrieve})      \\
\end{array}
\right.$$ Remark: The hypothetical oracle $\mathcal{O}$ is assumed
to have unlimited computing resources, and
 $\mathcal{S}_1$ always learns exactly the input related
to the request made by the attacker.

\section{Bringer-Chabanne EPIR Protocol  }

The EPIR protocols for testing equality and computing weighted
Hamming distance of \cite{BCPT07} are based on a  pre-processing
technique. Specifically, the user  sends an encryption of its input
$(F,i)$ to $\mathcal{DB}$, who then computes a temporary database
which contains  an encryption of $F(R_i)$. Finally, the user
executes a
 single-database CPIR protocol with $\mathcal{DB}$ to retrieve the encryption of $F(R_i)$.
This technique does not
allow the evaluation of  generic functions and incurs heavy
computation during the computation of the temporary database. The
Bringer-Chabanne EPIR protocol   aims to
 avoid these deficiencies. It is based on  ElGamal encryption schemes  over the multiplicative groups
 of finite fields.

\subsection{ElGamal Encryption Scheme}\label{sec:elgamal}

Let $p$ be a prime and $K={\rm GF}(p)$ be the finite field of order
$p$. Let $L={\rm GF}(p^n)$ be the
 finite field of order $p^n$  and $\mathbb{G}=L^{\times}$ be its  multiplicative group  of order $q=p^n-1$
for an integer $n\geq 2$. Let $g$ be a generator of $\mathbb{G}$.
The ElGamal encryption scheme over  $\mathbb{G}$ is a triplet of  algorithms $\Pi=({\bf Gen, Enc, Dec})$,
where
\begin{enumerate}
\item ${\bf Gen}$ is a key generation algorithm which takes as input a security parameter $1^k$ and proceeds as follows:
        \begin{enumerate}
             \item  generates the parameters $p,n,q$ and $ g$;
             \item  picks $x\leftarrow \mathbb{Z}_q$  and computes $y=g^x$;
            \item   outputs $pk=(q,g,y)$ as the public key and $sk=x$ as  the secret key.
       \end{enumerate}

\item $\bf Enc$ is an encryption algorithm which takes as input a plaintext $m\in \mathbb{G}$, picks
 $r\leftarrow\mathbb{Z}_q$  and outputs   $c=(g^r,y^rm)$ as the ciphertext.
\item   ${\bf Dec}$ is a decryption algorithm which  takes as input a ciphertext $c=(c_1,c_2)\in \mathbb{G}^2$ and outputs
$c_2\cdot c_1^{-x}$.
\end{enumerate}

\subsection{Requirements on Database Blocks and Functions}\label{sec:input_domain}

Following the notations  in Section \ref{sec:elgamal}, let
$\alpha\in L$ be a primitive element of the field extension $L/K$.
Then there is a  polynomial $G(t)\in K[t]$ of degree $<n$ such that
 $G(\alpha)=g$.
Let $x\in \mathbb{Z}_q$ and  $Y(t)\in  K[t]$ be the polynomial of degree $<n$ such that
$Y(\alpha)=y=g^x$.

For the Bringer-Chabanne EPIR protocol   to be correct, it is
required
 in \cite{BC09} that
 for every $j\in[N]$, the database block $R_j$ should belong to  $\mathbb{D}$, where
$$\mathbb{D}=\{\beta\in \mathbb{G}: Y(\beta)=G(\beta)^x{\rm ~and~} G(\beta)\neq 0\}.$$
The function to be evaluated by $\mathcal{U}$ can be any polynomial
over $L$, i.e., $\mathcal{F}=L[t]$.

\begin{figure}[t]
\begin{center}
\begin{boxedminipage}{12cm}
\begin{enumerate}
\itemsep=-0.1cm
\item $\mathcal{U}$:
Generates an ElGamal key pair $(pk, sk)$, where $pk = (q, g, y), y =
g^x$, and $sk = x$ is randomly chosen from $\mathbb{Z}_q$.
$\mathcal{U}$ also sends $pk$ to let $\mathcal{DB}$ the possibility
to verify the validity of $pk$ as an ElGamal public key. In
practice, the validity of $pk$ can be certified by a TTP, and the
same $pk$ can be used by the user for all his queries.

\item $\mathcal{U}$: For any polynomial function $F: {\rm GF}(p^n) \rightarrow {\rm GF}(p^n)$
and any index $1 \leq i \leq N$, computes $C_1,\ldots , C_N$ and
sends them to $\mathcal{DB}$ where
\begin{enumerate}
\item[-] $C_i = {\bf Enc} (F (\alpha) + r) = (G(\alpha)^{r_i} , Y (\alpha)^{r_i} (F (\alpha) + r))$
\item[-] and $C_j = {\bf Enc} (1) = (G(\alpha)^{r_j} , Y (\alpha)^{r_j} )$ for all  $j\neq i$,
\end{enumerate}
with randomly chosen $r\in {\rm GF}(p), r_j \in \mathbb{Z}_q (1\leq
j \leq N)$. Each $C_j$ can be written as $C_j = (V_j (\alpha),W_j
(\alpha))$ where $V_j$ and $W_j$ are polynomial over ${\rm GF}(p)$
of degree at most $n-1$.
\item $\mathcal{DB}$: After reception of the $C_j$, checks that they are nontrivial
ElGamal ciphertexts and computes $C_j (R_j) = (V_j (R_j),W_j (R_j))$
by replacing each occurrence of $\alpha$ (resp. $\alpha^{l}$ for
all power $l< n$) with $R_j$ (resp. with $R^{l}_j$).
\item  $\mathcal{DB}$: {Performs the product of all the $C_j$ together with a random encryption
of 1, say ${\bf Enc}(1) = (g^{r^{\prime}}, y^{r^{\prime}})$, sends
${\bf Enc}(1)\times \prod_{j=1}^{N}C_{j}(R_j)=\big( g^{r^{\prime}}\prod_{j=1}^{N}G(R_j)^{r_j},
y^{r^{\prime}}\big(\prod_{j=1}^{N}Y(R_j)^{r_j}\big)(F(R_i)+r)\big)$ to $\mathcal{U}$.}
\item $\mathcal{U}$:  Outputs ${\bf Dec}(sk, {\bf Enc}(1) \prod_{j=1}^{N}C_{j}(R_j))-r$ as $F(R_i)$.
\end{enumerate}
\end{boxedminipage}
\end{center}
\caption{Bringer-Chabanne EPIR protocol}
\label{fig:EPIR}
\end{figure}

\subsection{Bringer-Chabanne EPIR Protocol  }

Figure \ref{fig:EPIR} is the Bringer-Chabanne EPIR protocol, where  most notations are adopted from
 Section \ref{sec:elgamal} and Section \ref{sec:input_domain}.
The authors of the protocol expect to embed the description of the
polynomial $F(t)\in L[t]$ chosen by $\mathcal{U}$ into an ElGamal
ciphertext such that it can be evaluated by $\mathcal{DB}$ in an
oblivious way.

\newpage

The correctness of the Bringer-Chabanne EPIR protocol was claimed in \cite{BC09} as follows.
\begin{claim}\label{clm:correctness}
{\em (Section 4.4 of \cite{BC09})} A query {\em (say ${\bf
retrieve}(F,i)$)} gives the expected result
 {\em (i.e., $F(R_i)$)} as soon as there is no index $j$ for which one
of the values $G(R_j)$ or $Y(R_j)$ is zero, which may occur only
with a negligible probability in practice, leading to the
correctness of the EPIR protocol.
\end{claim}

\section{On the Incorrectness of Bringer-Chabanne EPIR Protocol}

In this section, we show that Bringer-Chabanne EPIR protocol does not satisfy the
correctness requirement defined in Section 2.2.   To simplify the
argument, we  give a  restricted version of Bringer-Chabanne EPIR protocol in which
$\mathcal{DB}$ is deterministic and  $N=1$. The restricted version
satisfies the correctness requirement as long as  Bringer-Chabanne EPIR protocol
 satisfies the correctness requirement. Then we turn to study  the incorrectness of the restricted version.

\subsection{Restricted Version}
At step (iv) of  the Bringer-Chabanne EPIR protocol, $\mathcal{DB}$ is
randomizing the product  $ \prod_{j=1}^{N}C_{j}(R_j)$ and sending
 ${\bf Enc}(1) \cdot \prod_{j=1}^{N}C_{j}(R_j)$ to the user.
We  note that the user could have computed the same output if $\mathcal{DB}$ merely
 sends  $ \prod_{j=1}^{N}C_{j}(R_j)$.
Therefore, we can safely modify step (iv) such that $\mathcal{DB}$
merely sends  $ \prod_{j=1}^{N}C_{j}(R_j)$  to $\mathcal{U}$
 with no impact on the correctness of the protocol. Let $i=N=1$.
  Then we have the restricted version (see Figure \ref{fig:RV}).

\begin{figure}[ghp]
\begin{center}
\begin{boxedminipage}{12cm}
\begin{enumerate}
\itemsep=-0.1cm
\item $\mathcal{U}$:
Generates an ElGamal key pair $(pk, sk)$, where $pk = (q, g, y), y =
g^x$, and $sk = x$ is randomly chosen from $\mathbb{Z}_q$.
$\mathcal{U}$ also sends $pk$ to let $\mathcal{DB}$ the possibility
to verify the validity of $pk$ as an ElGamal public key. In
practice, the validity of $pk$ can be certified by a TTP, and the
same $pk$ can be used by the user for all his queries.
\item $\mathcal{U}$: For any polynomial function $F: {\rm GF}(p^n) \rightarrow {\rm GF}(p^n)$,
computes $C = {\bf Enc} (F (\alpha) + r) = (G(\alpha)^{s} , Y
(\alpha)^{s} (F (\alpha) + r))$ and sends it to $\mathcal{DB}$ where
 $r\in {\rm GF}(p), s \in \mathbb{Z}_q$ are randomly chosen. The ciphertext $C$ can
be written as $C = (V(\alpha),W(\alpha))$ where $V$ and $W$ are
polynomials over ${\rm GF}(p)$ of degree at most $n-1$.
\item$\mathcal{DB}$: After reception of $C$ , checks that it is a  nontrivial
ElGamal ciphertext and computes $C(R) = (V(R),W(R))$ by replacing
each occurrence of $\alpha$ (resp. $\alpha^{l}$ for all power
$l< n$) with $R$ (resp. with $R^{l}$).
\item $\mathcal{DB}$: Sends $C(R)$ to $\mathcal{U}$.
\item $\mathcal{U}$:  Outputs  ${\bf Dec}(sk, C(R))-r$ as $F(R)$.
\end{enumerate}
\end{boxedminipage}
\end{center}
\caption{A restricted version of Bringer-Chabanne EPIR protocol}
\label{fig:RV}
\end{figure}

Clearly, if Claim \ref{clm:correctness} holds, then we have:
\begin{claim}\label{clm:correctness_restricted}
A query {\em (say ${\bf retrieve}(F,1)$)} in an execution of
the restricted version  gives
 $\mathcal{U}$ the expected result {\em (i.e., $F(R)$)} for any  $R\in \mathbb{G}$
 satisfying $Y(R)=G(R)^x$ and $G(R)\neq 0$.
\end{claim}

\subsection{Counterexample}\label{sec:counterexample}
We show that Claim \ref{clm:correctness_restricted} does not holds
by a counterexample. Let $p=2, n=3, K={\rm GF}(2), L={\rm GF}(2^3)$
and $\mathbb{G}=L^{\times}$. Let $\alpha=g\in \mathbb{G}$ be a
generator of $\mathbb{G}$ with minimal polynomial ${\rm
Min}_{g}(t)=t^3+t+1\in K[t]$. Figure \ref{fig:exeRV} is an  execution of
the restricted version  which does not give $\mathcal{U}$ the expected result.

\begin{figure}[ghp]
\begin{center}
\begin{boxedminipage}{12cm}
\begin{enumerate}
\itemsep=-0.1cm
\item $\mathcal{U}$: Picks a private key $sk=x=6\in \mathbb{Z}_7$, sets $y=g^2 + 1$ and $pk=(7,g,y)$.
$(pk,sk)$ is a pair of public and private keys for  the ElGamal
encryption scheme over group $\mathbb{G}$.
 $\mathcal{U}$ sends $pk$ to
$\mathcal{DB}$ such that $\mathcal{DB}$ can verify the validity of
$pk$ as an ElGamal public key. Clearly, $g=G(\alpha)$ and
$y=Y(\alpha)$ for polynomials $G(t)=t, Y(t)=t^2+1\in K[t]$ of degree
less
 than 3.
The field elements $R\in L$ which satisfy  equality $Y(R)=G(R)^x$
are $g,g^2$ and $g^2+g$.
\item  $\mathcal{U}$:  For a polynomial function $F(t)=g\in L[t]$, takes $s=6\in \mathbb{Z}_7, r=1\in K$
and computes the ciphertext $C = {\bf Enc} (F (\alpha) + r) =
(G(\alpha)^{s} , Y (\alpha)^{s} (F (\alpha) + r))=(g^6,
(g^2+1)^6(g+1))=(g^2+1,g^2+g)$ and sends it to $\mathcal{DB}$.
Clearly,  we have that $V(t)=t^2+1$ and $W(t)=t^2+t$.
\item  $\mathcal{DB}$: Sets the database block to be $R=g^2+g\in \mathbb{G}$.  After receiving the
 ciphertext $C=(g^2+1,g^2+g)$ from $\mathcal{U}$, $\mathcal{DB}$ checks that $C$ is a  nontrivial
ElGamal ciphertext and computes $C(R) =
(V(R),W(R))=(R^2+1,R^2+R)=(g+1,g^2)$ by replacing each occurrence of
$\alpha$ (resp. $\alpha^{l}$ for all power $l< n$) with $R$
(resp. with $R^{l}$).
\item  $\mathcal{DB}$: Sends $C(R)=(g+1,g^2)$ to $\mathcal{U}$.
\item $\mathcal{U}$: Outputs ${\bf Dec}(sk,C(R))-r=g^2+g$ as $F(R)$, which is absurd (since $F(R)=g$).
\end{enumerate}
\end{boxedminipage}
\end{center}
\caption{An execution of the restricted version}
\label{fig:exeRV}
\end{figure}

\subsection{Failure Probability}\label{sec:notations}

We have seen that the restricted version  may not give $\mathcal{U}$ the expected
result in Section \ref{sec:counterexample}. However, given the counterexample, we cannot conclude that
the Bringer-Chabanne EPIR protocol does not satisfy the correctness requirement defined in
Section \ref{sec:security_model}.
In fact, an EPIR protocol is said to be correct as long as it always gives $\mathcal{U}$ the expected
result for any fixed input $(F(t),i)\in L[t]\times [n]$ {\em except with a negligible probability}.
In other words, as a collection of probabilistic algorithms,
an EPIR protocol is {\em allowed to fail with a negligible probability}.
Therefore, to show that the Bringer-Chabanne EPIR protocol
does not satisfy the correctness requirement, it is necessary to compute the failure probability
of the protocol, i.e.,
the probability that
the protocol  does not give $\mathcal{U}$ the expected result.

In this section, we  study the
failure  probability of the restricted version. We show,
through  experimental results,  that the restricted version  does fail with large
probability for certain choices of $F(t)$ (e.g.,  $F(t)=g$).

From now on, we fix $p=2$ to be the characteristic of all related
finite fields. However, we stress that our methodology is applicable
to any characteristic
 $p$.
Following the notations of Section \ref{sec:elgamal} and Section 
\ref{sec:input_domain}, let $K={\rm GF}(2)$  and
$L={\rm GF}(2^n)$ be the extension of $K$ of degree $n$ for an integer  $n\geq 2$.
  Let $\mathbb{G}=L^{\times}$ be the multiplicative group
of $L$ of order $q=2^n-1$ and $g$ be a generator of $\mathbb{G}$.
W.l.o.g., we suppose $\alpha=g$. Then $G(t)=t\in K[t]$ is the
polynomial of degree less than $n$ such that $G(\alpha)=g$.
 For every $x\in \mathbb{Z}_q$, let $Y(t)\in K[t]$ be the polynomial of degree less than $n$ such that
  $Y(\alpha)=y=g^x$.
We define
\begin{equation*}\label{eqn:D_t}
D(t)=G(t)^x+Y(t)=t^x+Y(t)\in K[t].
\end{equation*}
Then the set of database blocks which satisfy the requirements
imposed by Claim \ref{clm:correctness_restricted} (or in Section
\ref{sec:input_domain})  is
\begin{equation*}\label{eqn:D_ngx}
\mathbb{D}_{n,g,x}=\{\beta\in \mathbb{G}| D(\beta)=0\}.
\end{equation*}

We say that an execution of the restricted version  is {\em parameterized} by
$(n,g,x,F,\\s,r,R)$ if $x\in \mathbb{Z}_q, F(t)\in L[t], s\in
\mathbb{Z}_q, r\in K$ and $R\in \mathbb{D}_{n,g,x}$ are the private
key, the polynomial to be evaluated,  the randomness used at step (ii) of the restricted version 
and  the database block held by $\mathcal{DB}$,  respectively.
 Let $V(t), W(t)\in K[t]$ be the polynomials of degree less than $n$
such that $V(g)=g^s$ and $W(g)=y^s(F(g)+r)$. Then  the
execution of the restricted version  parameterized by $(n,g,x,F,s,r,R)$ gives
$\mathcal{U}$ the expected result {if and only if} $V(R)\neq 0$ and
$E(R)=0$, where
\begin{equation}\label{eqn:Q_t}
E(t)=W(t)+V(t)^x(F(t)+r).
\end{equation}
For an  execution of  the restricted version   parameterized by
$(n,g,x,F,s,r,R)$,  we define
\begin{equation*}\label{eqn:H_xsrFR}
{\bf H}_{x,s,r,{\scriptscriptstyle  F, R}}=
\begin{cases}
1 & \text{if $V(R)\neq 0$ and $E(R)=0$,}
\\
0 &\text{otherwise.}
\end{cases}
\end{equation*}
Then the execution  fails if and only if ${\bf H}_{x,s,r,{\scriptscriptstyle  F,
R}}=0$. Therefore, the probability that an execution of the restricted version
fails when $x\in \mathbb{Z}_q$ is the private key and  $F(t)\in
L[t]$  is the polynomial
 chosen by $\mathcal{U}$ is exactly
\begin{equation*}\label{eqn:epsilon}
\epsilon(n,g,x,F)=\Pr\left[s\leftarrow \mathbb{Z}_q, r\leftarrow  K, R\leftarrow
\mathbb{D}_{n,g,x}: {\bf H}_{x,s,r,{\scriptscriptstyle  F, R}}=0\right].
\end{equation*}
Since $s,r$ and $R$ are uniformly distributed, we have that
\begin{equation}\label{eqn:epsilon}
\epsilon(n,g,x,F)=\frac{\displaystyle\sum_{s\in
\mathbb{Z}_q}\sum_{r\in K} \sum_{R \in \mathbb{D}_{n,g,x}} (1-{\bf
H}_{x,s,r,{\scriptscriptstyle  F,R}})}{2q\cdot  |\mathbb{D}_{n,g,x}|}.
\end{equation}
The probability that the restricted version   fails when $F(t)\in L[t]$ is the
polynomial
 chosen by $\mathcal{U}$
is exactly
\begin{equation}\label{eqn:eta}
\eta(n,g,F)=\frac{1}{q} \sum_{x\in \mathbb{Z}_q}\epsilon(n,g,x,F).
\end{equation}

The   probabilities $\eta(n,g,F)$ for $2\leq n\leq 9$ and $F(t)=g$
are quite large and  enumerated in Table \ref{table:failprob}.
\begin{table}[ghp]
\begin{center}
\begin{tabular}{|c|c|c||c|c|c|}\hline
$n$ & ${\rm Min}_g(t)$   & $\eta(n,g,g)$  & $n$ & ${\rm Min}_g(t)$   & $\eta(n,g,g)$
\\ \hline \hline
2 & $t^2+t+1$    &0.61111    & 6 & $t^6+t^4+t^3+t+1$          & 0.87719         \\ \hline
3 & $t^3+t+1$    & 0.74271   & 7 & $t^7+t+1$                        & 0.87895   \\ \hline
4 & $t^4+t+1$    & 0.81537   & 8 & $t^8+t^4+t^3+t^2+1$      & 0.89809           \\ \hline
5 & $t^5+t^2+1$  & 0.83630   & 9 & $t^9+t^4+1$                    & 0.90358     \\ \hline
\end{tabular}
\end{center}
\caption{Failure probability}
\label{table:failprob}
\end{table}

\subsection{Bringer-Chabanne EPIR Protocol Fails Frequently When $F(t)=g$}\label{sec:BC_EPIR_fails_frequently_for_g}

In this section, we show that the restricted version  fails with large
probability when $F(t)=g$. Specifically, for every integer $n\geq
2$, we give lower bound on $\eta(n,g,g)$.

We follow the notations in Section \ref{sec:notations}. For every
$j\in \mathbb{Z}_q$, the set  ${\bf C}=\{j\cdot 2^k\bmod q|
k=0,1,2,\cdots \}$ is called a {\em cyclotomic coset} mod $q$.   By
default,  ${\bf C}$ is  represented by the smallest  number $u\in
{\bf C}$ and denoted as
\begin{equation*}
{\bf C}_u= \left\{j\cdot 2^k\bmod q| k=0,1,2,\cdots \right\}.
\end{equation*}
The number $u$ is called the {\em coset representative} of ${\bf
C}$. Clearly, all distinct cyclotomic cosets mod $q$ are pairwise
disjoint and form a partition of $\mathbb{Z}_q$, i.e.,
$\mathbb{Z}_q=\bigcup_{u\in U} {\bf C}_u$, where $U$ is the set of
coset representatives of all distinct cyclotomic cosets mod $q$. For
every positive integer $d$, we denote by $N_2(d)$ the  number of
monic irreducible polynomials of degree $d$ in $K[t]$.
\begin{lemma}\label{lem:cyc}
{\em  (Lidl and Niederreiter \cite{LN97})}
The following statements hold:
\begin{enumerate}
\item For every $u\in U$, the cardinality of ${\bf C}_u$ is a divisor of
$n$.
\item For every positive integer $d|n$, the number of cyclotomic cosets
mod $q$ of cardinality $d$ is $N_2(d)$.
\item  For  every integer $d\geq 2$, we have that
$\displaystyle N_2(d)\leq \frac{1}{d}(2^d-2).$
\end{enumerate}
\end{lemma}
For every $u\in U$, we denote by
\begin{equation*}
{\bf D}_u=\left\{g^j| j\in {\bf C}_u\right\}
\end{equation*}
 the set of field elements  in $L$ which share the same {\em minimal polynomial} over $K$ with  $g^u$.
For every $x\in \mathbb{Z}_q$, it is clear that there  is  a subset
$U_x\subseteq U$ of coset representatives such that
\begin{equation}\label{eqn:database_elt}
\mathbb{D}_{n,g,x}=\bigcup_{u\in U_x}{\bf D}_u.
\end{equation}
\begin{lemma}\label{lem:Unotempty}
 For every  $x\in \mathbb{Z}_q$,  we have that $1 \in U_x$.
\end{lemma}

\begin{proof}
It follows from the fact that $D(t)\in K[t]$ and $D(g)=0$.
\end{proof}

Due to (\ref{eqn:Q_t}), $E(t)$ is  determined by the parameters
$g\in \mathbb{G},x\in\mathbb{Z}_q, F(t)\in L[t], s\in \mathbb{Z}_q$
and $r\in K$. Next lemma shows that $E(t)$ and $D(t)$ only share a
very
 small number of roots in $L$
when $F(t)=g$.
\begin{lemma}\label{lem:rootQ}
Suppose  $F(t)=g$. Then for every $x\in \mathbb{Z}_q, u\in U_x, s\in
\mathbb{Z}_q$ and $ r\in K$, either $V(\beta)=0$ for every $\beta\in
{\bf D}_u$ or $E(t)$ has at most one root in ${\bf D}_u$.
\end{lemma}
\begin{proof}
If $V(g^u)=0$, then $V(g^{2^j\cdot u})=V(g^u)^{2^j}=0$ for any
$j\in \mathbb{N}$, i.e.,
 $V(\beta)=0$ for every $\beta\in {\bf D}_u$. Otherwise, we show that $E(t)$ has at most one root in
 ${\bf D}_u$.
Due to  (\ref{eqn:Q_t}), we have that
$$E(t)=W(t)+V(t)^x(g+r).$$
 Suppose that $E(t)$ has two different roots in  ${\bf D}_u$, say $g^{u\cdot 2^j}$ and $g^{u\cdot 2^k} $, where
 $0\leq j<k <n$. Then
$$W(g^{u\cdot 2^j})+V(g^{u\cdot 2^j})^x(g+r)=0=W(g^{u\cdot 2^k})+V(g^{u\cdot 2^k})^x(g+r).$$
It follows that
$$(g+r)^{2^{n-j}}=(W(g^u)/V(g^u)^x)^{2^n}=(g+r)^{2^{n-k}}.$$
Since $r\in K$, the above equality implies
$g^{2^{n-j}}=g^{2^{n-k}}.$  Since
$g$ is primitive, we have $(2^n-1)|(2^{n-j}-2^{n-k})$. It follows that
 $n|(k-j)$, which is a
contradiction.
\end{proof}

The following lemma gives  lower bound on $\epsilon(n,g,x,g)$ for
any private key $x\in \mathbb{Z}_q$.

\begin{lemma}\label{lem:lower_bound_epsilon}
For every $x\in \mathbb{Z}_q$, we have that $\displaystyle
\epsilon(n,g,x,g)\geq 1-\frac{|U_x|}{|\mathbb{D}_{n,g,x}|}.$
\end{lemma}
\begin{proof}
Due to  (\ref{eqn:epsilon}) and
(\ref{eqn:database_elt}), we have that
\begin{eqnarray*}
\epsilon(n,g,x,g)&=&\frac{\displaystyle\sum_{s\in
\mathbb{Z}_q}\sum_{r\in K}\sum_{R \in \mathbb{D}_{n,g,x}} (1-{\bf
H}_{x,s,r,g,{\scriptscriptstyle  R}})}{2q\cdot |\mathbb{D}_{n,g,x}|}\\
&=&
\frac{\displaystyle\sum_{s\in \mathbb{Z}_q}\sum_{r\in K}\sum_{u\in
U_x}\sum_{R\in {\bf D}_u}
 (1-{\bf H}_{x,s,r,g,{\scriptscriptstyle  R}})}{2q\cdot |\mathbb{D}_{n,g,x}|}.
\end{eqnarray*}
Let $s\in \mathbb{Z}_q$ and $r\in K$ be arbitrary.  Due to Lemma
\ref{lem:rootQ}, for every $u\in U_x$, either $V(\beta)=0$ for every
$\beta\in {\bf D}_u$, or $E(t)$ has at most one root in ${\bf D}_u$.
It follows that
$$\sum_{R\in {\bf D}_u} (1-{\bf H}_{x,s,r,g,{\scriptscriptstyle  R}})\geq |{\bf C}_u|-1.$$
Therefore,
\begin{eqnarray*}
\epsilon(n,g,x,g)\geq \frac{\displaystyle\sum_{s\in
\mathbb{Z}_q}\sum_{r\in K}\sum_{u\in U_x}(|{\bf C}_u|-1)} {2q
\cdot |\mathbb{D}_{n,g,x}|}=
 1-\frac{|U_x|}{|\mathbb{D}_{n,g,x}|}.
\end{eqnarray*}
\end{proof}

We want to bound $\epsilon(n,g,x,g)$ for various settings of $n$ and $x$.
As the first case, we suppose that $n$ is a prime and have the following lemma:

\begin{lemma}\label{lem:lower_bound_eta_for_prime_n}
If $n$ is  prime, then  $\displaystyle \epsilon(n,g,x,g)> 1-\frac{2}{n}$ for every $x\in \mathbb{Z}_q$
\end{lemma}

\begin{proof}
Due to Lemma \ref{lem:cyc}, $|{\bf C}_u|$ divides $n$ for every
$x\in \mathbb{Z}_q$ and $ u\in U_x$. Since $n$ is prime, we have
that $| {\bf C}_u|=1$ or $n$.
\begin{enumerate}
\item If $|U_x|=1$, then $U_x=\{1\}$ due to Lemma \ref{lem:Unotempty}. It is obvious that $|{\bf C}_1|=n$. By Lemma \ref{lem:lower_bound_epsilon},  we have
$$\epsilon(n,g,x,g)\geq1-\frac{|U_x|}{|\mathbb{D}_{n,g,x}|}=1-\frac{1}{n}>1-\frac{2}{n}.$$
\item If $| U_x|>1$ and $0\in U_x$, then we have that
$$\epsilon(n,g,x,g)\geq 1-\frac{|U_x|}{|\mathbb{D}_{n,g,x}|}=1-\frac{|U_x|}{1+n(|U_x|-1)}>1-\frac{2}{n}.$$
\item If $|U_x|>1$ and $0\notin U_x$, then we have that
$$\epsilon(n,g,x,g)\geq 1-\frac{|U_x|}{|\mathbb{D}_{n,g,x}|}=
1-\frac{|U_x|}{n\cdot |U_x|}
=1-\frac{1}{n}>1-\frac{2}{n}.$$
\end{enumerate}
\end{proof}

Below we lower bound $\epsilon(n,g,x,g)$ for {\bf any} integer $n\geq 2$ and  private key $x\in
\mathbb{Z}_q$.
 For any positive integer $d|n$, we
set
 \begin{equation*}
\lambda_{x.d}=|\{u: u\in U_x~{\rm and~} {\bf C}_u {\rm~is~of~cardinality~}d\}|.
 \end{equation*}
Due to Lemma \ref{lem:Unotempty} and the requirements on database
block $R$ (imposed by Claim
 \ref{clm:correctness_restricted}),  $\lambda_x=(\lambda_{x,d})$ belongs to the following set
\begin{equation*}\label{eqn:domain}
\begin{split}
\Psi_n=\{z=(z_d)_{d|n}: 0\leq z_1\leq 1; 1\leq z_n\leq
N_2(n); \\ 0\leq z_d \leq N_2(d) {\rm~for~} d|n, 1<d<n\},
\end{split}
\end{equation*}
where the coordinates of $\lambda_x$ and $z$ are  indexed by
positive divisors of $n$. Due to Lemma
\ref{lem:lower_bound_epsilon}, we have that
\begin{equation}\label{eqn:lower_bound_epsilon_any_x}
\epsilon(n,g,x,g)\geq 1-\frac{|U_x|}{|\mathbb{D}_{n,g,x}|}=
 1-\frac{\sum_{d|n}\lambda_{x,d}}
{\sum_{d|n}d\lambda_{x,d}}.
\end{equation}
We turn to  upper bound the following function
\begin{eqnarray*}
\psi_n(z)=\frac{\sum_{d|n}z_d}{\sum_{d|n}dz_d},
\end{eqnarray*}
on $\Psi_n$. Because this  is relatively hard, we turn to upper
bound the function
\begin{eqnarray*}
\phi_n(z)=\frac{\sum_{d=1}^nz_d}{\sum_{d=1}^ndz_d},
\end{eqnarray*}
where  $z=(z_1,\ldots, z_n)$ is taken from the following set
\begin{equation*}\label{eqn:domain}
\begin{split}
\Phi_n=\{z=(z_1,\ldots,z_n): 0\leq z_1\leq 1; 1\leq z_n\leq
N_2(n); \\ 0\leq z_d \leq N_2(d) {\rm~for~} 1<d<n\}.
\end{split}
\end{equation*}
Let $\omega(n)$ be the maximum value of $\phi_n(z)$ on $\Phi_n$,
i.e.,
\begin{equation*}
\omega(n)=\max\{\phi_n(z): z\in \Phi_n\}.
\end{equation*}

\begin{lemma}\label{lem:lower_bound_large_x}
For every $x\in \mathbb{Z}_q$, we have that $\epsilon(n,g,x,g)\geq
1-\omega(n)$.
\end{lemma}
\begin{proof}
Clearly, $\omega(n)=\max\{\phi_n(z): z\in \Phi_n\} \geq
\max\{\psi_n(z): z\in \Psi_n\}\geq \psi_n(\lambda_x)$. Due to
 (\ref{eqn:lower_bound_epsilon_any_x}), we have that
$\epsilon(n,g,x,g)\geq 1-\psi_n(\lambda_x)\geq  1-\omega(n)$ for
every $x\in \mathbb{Z}_q$.
\end{proof}

Due to Lemma \ref{lem:lower_bound_large_x}, it is sufficient  to
upper bound $\omega(n)$.
\begin{lemma}\label{lem:omega_maximality}
Suppose that $\omega(n)=\phi_n(\xi)$ for
$\xi=(\xi_1,\ldots,\xi_n)\in \Phi_n$. Then  $\xi_1= \xi_n=1$.
Furthermore, if $n\geq 3$, then there is an integer $1<h<n$ such
that $\xi_d=N_2(d)$ for every integer $1<d\leq h$ and $\xi_d=0$ for
every integer     $h<d<n$.
\end{lemma}
\begin{proof}
It is trivial to verify that $\xi_1=\xi_2=1$ for $n=2$. Let $n\geq 3$.

\begin{enumerate}
\item  For every $(0,z_2,\ldots,z_n), (1,z_2,\ldots,z_n)\in \Phi_n$, it is easy to see that
\begin{eqnarray*}
\phi_n(0, z_2,\ldots, z_n)-\phi_n(1,z_2,\ldots,
z_n)<0,
\end{eqnarray*}
which implies that  $\xi_1=1$.

\item For every $(1,z_2,\ldots, z_{n-1}, z_n), (1,z_2,\ldots, z_{n-1},1)\in
\Phi_n$ (where $z_n >1$), it is easy to see that
\begin{eqnarray*}
\phi_n(1,z_2, \ldots, z_{n-1}, z_n)-\phi_n(1,z_2,\ldots,z_{n-1}, 1)
 < 0,
\end{eqnarray*}
which implies that $\xi_n=1$.

\item Suppose  $0<\xi_h<N_2(h)$ for some integer $1<h<n$.  Let
$$C_1=\sum_{d=1}^{h-1}\xi_d, C_2=\sum_{d=h+1}^n\xi_d, C_3=\sum_{d=1}^{h-1}d\xi_d,
C_4=\sum_{d=h+1}^nd\xi_d.$$ Then due to the maximality of
$\omega(n)$, we have that
\begin{eqnarray*}\label{eqn:det_h_1}
0&\geq& \phi_n(\xi_1,\ldots, \xi_h+1, \ldots, \xi_n)-\phi_n(\xi)\\
 &=&
\frac{C_3+C_4-hC_1-hC_2}{(C_3+h(\xi_h+1)+C_4)(C_3+h\xi_h+C_4)};
\end{eqnarray*}
\begin{eqnarray*}\label{eqn:det_h_2}
0&\geq& \phi_n(\xi_1,\ldots, \xi_h-1, \ldots, \xi_n)-\phi_n(\xi) \\
&=&
\frac{-C_3-C_4+hC_1+hC_2}{(C_3+h(\xi_h-1)+C_4)(C_3+h\xi_h+C_4)}.
\end{eqnarray*}
The above inequalities  imply that
$C_3+C_4=hC_1+hC_2$. Hence, we have
 \begin{eqnarray*}\label{eqn:value_h}
h=\frac{\sum_{d=1}^n d\xi_d}{\sum_{d=1}^n
\xi_d}=\frac{1}{\omega(n)}.
 \end{eqnarray*}

\item We claim that $\xi_a=N_2(a)$ for every $1<a<h$. Otherwise, by (iii), we have that $\xi_a=0$ and
\begin{eqnarray*}
\omega(n)<\phi_n(\xi_1,\hdots, \xi_{a}+1, \hdots, \xi_h-1, \hdots,
\xi_n),
\end{eqnarray*}
which is a contradiction.

\item We claim  that $\xi_b=0$ for every $h<b<n$. Otherwise, by (iii), we have that $\xi_b=N_2(b)$ and
\begin{eqnarray*}
\omega(n)< \phi_n(\xi_1,\hdots, \xi_h+1, \hdots, \xi_b-1, \hdots, \xi_n),
\end{eqnarray*}
which is a contradiction.

\item
Finally, we show that  $\omega(n)=\phi_n(1,N_2(2),\ldots, N_2(h),
0,\ldots,0,1)$. Due to (iii), (iv) and (v), we have that
$$\xi=(1,N_2(2),\ldots, N_2(h-1),\xi_h, 0,\ldots,0,1),$$
where $0<\xi_h<N_2(h)$.
Since $\phi_n(\xi)=\omega(n)\geq \phi_n(1,N_2(2),\ldots, N_2(h-1),0,
0,\ldots,0,1)$, we have
$$hC_1-C_3\leq n-h.$$  If $hC_1-C_3<  n-h$, then
\begin{eqnarray*}
\omega(n)<\phi_n(1,N_2(2),\ldots, N_2(h), 0,\ldots,0,1),
\end{eqnarray*}
which is a contradiction. Therefore, $hC_1-C_3=  n-h$.Then it is not
hard to verify that
$$\omega(n)=\phi_n(\xi)=
\phi_n(1, N_2(2), \ldots, N_2(h), 0,\ldots,0,1).$$ Therefore,  we
could have taken $\xi=(1, N_2(2), \ldots, N_2(h), 0,\ldots,0,1)$.
\end{enumerate}
\end{proof}

Due to Lemma  \ref{lem:omega_maximality}, for every integer $n\geq
3$, there is at least one integer $1<h<n$ such that
\begin{equation}\label{eqn:index_h}
\omega(n)=\phi_n(1,N_2(2),\ldots, N_2(h), 0,\ldots, 0,1).
\end{equation}
Note that the integer $h$ may be not unique. For every integer
$n\geq 3$, we define
\begin{equation}\label{eqn:define_h_n}
\begin{split}
h(n)=\min\{h: \omega(n)=
\phi_n(1,N_2(2),\ldots, N_2(h), 0,
\ldots, 0,1), \\ {\rm ~where~} 1<h<n\}
\end{split}
\end{equation}
to be the smallest integer $1<h<n$ such that
(\ref{eqn:index_h}) holds.
 Next lemma shows that $h(n)$ is an increasing  function of $n$.
\begin{lemma}\label{lem:h_n_is_increasing}
We have that $h(n+1)\geq h(n)$ for every  integer $n\geq 3$.
\end{lemma}
\begin{proof}
Due to the definition of $h(\cdot)$ by
(\ref{eqn:define_h_n}), it is not hard to see that
\begin{equation*}
\begin{split}
\phi_n(1,N_2(2),\ldots,N_2(l-1), N_2(l), 0,\ldots,0,1)>\\
\phi_n(1,N_2(2),\ldots, N_2(l-1),0, 0,\ldots,0,1)
\end{split}
\end{equation*}
 for every
integer $2\leq l\leq h(n)$. Equivalently, we have that
\begin{equation}\label{eqn:n_implication}
\frac{1}{l}>\frac{
\sum_{d=2}^{l-1}N_2(d)+2}{\sum_{d=2}^{l-1}dN_2(d)+n+1}
\end{equation}
for every integer $2\leq l\leq h(n)$. Due to 
(\ref{eqn:n_implication}), it is not hard to verify that
\begin{equation}\label{eqn:det_h_n_1}
\begin{split}
\phi_{n+1}(1,N_2(2),\ldots, N_2(l-1), N_2(l), 0,\ldots,0,1)>\\
\phi_{n+1}(1,N_2(2),\ldots, N_2(l-1),0,0,\ldots,0,1)
\end{split}
\end{equation}
for every integer $2\leq l\leq h(n)$. In particular, 
(\ref{eqn:det_h_n_1}) holds for $l=h(n)$. This implies that $h(n+1)\geq h(n)$.
\end{proof}

On the other hand, $\omega(n)$ is a decreasing function of $n$:
\begin{lemma}\label{lem:omega_n_is_decreasing}
We have that $\omega(n+1)< \omega(n)$ for every  integer
$n\geq 3$.
\end{lemma}
\begin{proof}
By Lemma \ref{lem:h_n_is_increasing}, we have that  $h(n+1)\geq
h(n)$.  If $h(n+1)=h(n)$, then
\begin{eqnarray*}
\omega(n+1)&=&\frac{\sum_{d=2}^{h(n+1)}N_2(d)+2}{\sum_{d=2}^{h(n+1)}dN_2(d)+n+2}
=\frac{\sum_{d=2}^{h(n)}N_2(d)+2}{\sum_{d=2}^{h(n)}dN_2(d)+n+2} \\
&<&
\frac{\sum_{d=2}^{h(n)}N_2(d)+2}{\sum_{d=2}^{h(n)}dN_2(d)+n+1}=\omega(n).
\end{eqnarray*}
If $h(n+1)>h(n)$, then
\begin{eqnarray*}
\omega(n)&=&\frac{\sum_{d=2}^{h(n)}N_2(d)+2}{\sum_{d=2}^{h(n)}dN_2(d)+n+1}\geq
\frac{1}{h(n)+1} \geq
\frac{1}{h(n+1)}\\
&>&
\frac{\sum_{d=2}^{h(n+1)}N_2(d)+2}{\sum_{d=2}^{h(n+1)}dN_2(d)+n+2}=
\omega(n+1),
\end{eqnarray*}
where the first and third inequalities follow from the
definition of   $h(\cdot)$ by 
(\ref{eqn:define_h_n}).
\end{proof}

We enumerate  the values of $h(n)$ and $\omega(n)$ for some integers $n$  in Table \ref{table:h_n}.

\begin{table}[t]
\begin{center}
    \begin{tabular}{|c|c|c||c|c|c||c|c|c|}\hline
$n$   & $h(n)$ & $\omega(n)$ & $n$  & $h(n)$ & $\omega(n)$ & $n$  & $h(n)$ & $\omega(n)$ \\ \hline\hline
2 &   1    & 0.66667    &  12 &   4    & 0.24242  & 296 &  10    & 0.09996   \\ \hline
3 &   1    & 0.50000    &  20 &   5    & 0.19718  & 522 &  11 & 0.09089\\ \hline
4 &   2    & 0.42857    &  34 &   6    & 0.16547  & 934 &  12 & 0.08332\\ \hline
5 &   2    & 0.37500    &  57 &   7    & 0.14236  & 1681 &  13 & 0.07692\\ \hline
6 &   2    & 0.33333    &  98 &   8    & 0.12478  & 3058 &  14 & 0.07143\\ \hline
7 &   3    & 0.31250    &  169 &   9    & 0.11101 & 5596 &  15 & 0.06667\\ \hline
     \end{tabular}
    \end{center}
    \caption{The values of $h(n)$ and $\omega(n)$}
    \label{table:h_n}
\end{table}

\begin{lemma}\label{lem:compare_for_small_x}
For every integer $n\geq 7$,
 we have that $\displaystyle \omega(n)\geq  \frac{5}{n+9}$.
\end{lemma}
\begin{proof}
Due to Table 2 and Lemma \ref{lem:h_n_is_increasing}, we have that
$h(n)\geq 3$ for every integer $n\geq 7$. It follows that
$\omega(n)\geq \phi_n(1,1,2,0,\ldots,0,1)=5/(n+9).$
\end{proof}

At last, we have the following theorem.
\begin{theorem}\label{thm:lower_bound_eta_for_any_n}
We have that
\begin{equation*}
\eta(n,g,g)\geq
\begin{cases}
1-\omega(n) & \text{{\rm if~}  $2\leq n\leq 6$ {\rm or $n\geq 7$ is composite;}} \\
1-\frac{2}{n} & \text{{\rm if~}  $n\geq 7$ {\rm is prime}.}
\end{cases}
\end{equation*}
\end{theorem}

\begin{proof}
Table \ref{table:h_n} shows that  $\omega(n)\leq 2/n$ for every integer $2\leq n\leq 6$.
Due to Lemma \ref{lem:lower_bound_eta_for_prime_n} and
Lemma \ref{lem:lower_bound_large_x},
we have that $\epsilon(n,g,x,g)\geq \max\{1-2/n,1-\omega(n)\}=1-\omega(n)$
for $n=2,3,5$,  and $\epsilon(n,g,x,g)\geq 1-\omega(n)$ for $n=4,6$.
Due to  (\ref{eqn:eta}), we have that
$$\eta(n,g,g)=\frac{1}{q}\sum_{x\in \mathbb{Z}_q}\epsilon(n,g,x,g)\geq 1-\omega(n).$$

Due to Lemma \ref{lem:lower_bound_eta_for_prime_n}, Lemma \ref{lem:lower_bound_large_x}
and Lemma \ref{lem:compare_for_small_x},
we have that $\epsilon(n,g,x,g)\geq \max\{1-2/n, 1-\omega(n)\}=
1-2/n$ if $n\geq 7$ is prime
and
$\epsilon(n,g,x,g)\geq 1-\omega(n)$ if $n\geq 7$ is composite.
 Due to  (\ref{eqn:eta}), we have that
 $\eta(n,g,g)\geq 1-2/n$ if $n\geq 7$ is prime and
  $\eta(n,g,g)\geq 1-\omega(n)$ if $n\geq 7$ is composite.
\end{proof}

By Theorem \ref{thm:lower_bound_eta_for_any_n}, Lemma \ref{lem:omega_n_is_decreasing}
and Table \ref{table:h_n},  we see that
$\eta(n,g,g)$ is always non-negligible.
Hence, we have the following theorem.
\begin{theorem}\label{thm:BC_is_not_correct_for_g}
The restricted version  does not satisfy the correctness requirement if $F(t)=g$.
\end{theorem}

\subsection{Extension to A Set of Polynomials}

In this section, we extend Theorem \ref{thm:BC_is_not_correct_for_g}
to a set of polynomials
 $F(t)\in L[t]$. In particular, we follow the notations in Section \ref{sec:BC_EPIR_fails_frequently_for_g} and
  show that the restricted version   does not satisfy the correctness requirement
if  $F(t)\in \mathcal{P}$, where
\begin{center}
$ \mathcal{P}=\{f(t)=\sum_{k=0}^{d}f_kt^k: \exists~ 0\leq l\leq
d{\rm ~such~that~}f_{l}\in L {\rm~is~primitive}
{\rm~and~}f_k\in K{\rm ~for~every~}k\neq l\}. $
\end{center}
Note that the polynomial $F(t)=g\in L[t]$ we studied in Section
\ref{sec:BC_EPIR_fails_frequently_for_g} is in
$\mathcal{P}$ and satisfies Lemma \ref{lem:rootQ}, which is critical
for obtaining all subsequent  lemmas and theorems.
 Next lemma shows that Lemma \ref{lem:rootQ} holds for any polynomial
$F(t)\in \mathcal{P}$ as well.
\begin{lemma}\label{lem:F_t_in_POLY_are_good}
Let  $F(t)\in \mathcal{P}$. Then for every $x\in \mathbb{Z}_q, u\in
U_x, s\in \mathbb{Z}_q$ and $ r\in K$, either $V(\beta)=0$ for every
$\beta\in {\bf D}_u$ or $E(t)$ has at most one root in ${\bf D}_u$.
\end{lemma}

\begin{proof}
If $V(g^u)=0$, then $V(g^{u\cdot 2^j})=V(g^u)^{2^j}=0$ for every
$j\in \mathbb{N}$, i.e., $V(\beta)=0$ for every $\beta\in {\bf
D}_u$. Otherwise, we have $V(\beta)\neq 0$ for every $\beta\in {\bf
D}_u$. Suppose
 $F(t)=\sum_{k=0}^d F_k t^k$, where $F_{l}\in L$ is of order $q$ and
 $F_k\in K$ for every $k\neq l$. We show that $E(t)$ has at most one root in ${\bf D}_u$, where
 $$E(t)=W(t)+V(t)^x(F(t)+r).$$
Suppose $E(t)$ has two different roots in ${\bf D}_u$, say $g^{u\cdot 2^a}$ and $g^{u\cdot 2^b}$,
where $0\leq a<b<n$. Then
$$W(g^{u\cdot 2^a})+V(g^{u\cdot 2^a})^x(F(g^{u\cdot 2^a})+r)
=0=W(g^{u\cdot 2^b})+V(g^{u\cdot 2^b})^x(F(g^{u\cdot 2^b})+r).$$ It
follows that
\begin{equation}\label{eqn:raise_to_high_power}
(F(g^{u\cdot 2^a})+r)^{2^{n-a}}=(F(g^{u\cdot 2^b})+r)^{2^{n-b}}.
\end{equation}
Let $c\in \{a,b\}$. Then it is not hard to see that
$$(F(g^{u\cdot 2^c})+r)^{2^{n-c}}=
\sum_{k=0}^{l-1}F_k g^{uk}+\sum_{k=l+1}^{d}F_k
g^{uk}+F_{l}^{2^{n-c}}g^{ul}+r.$$ Due to
(\ref{eqn:raise_to_high_power}), we have that
$F_{l}^{2^{n-a}}=F_{l}^{2^{n-b}}.$
Since $F_{l}\in L$ is primitive, we have
$(2^n-1)|(2^{n-a}-2^{n-b})$ and therefore $n|(b-a)$, which  is a
contradiction.
\end{proof}

Due to Lemma \ref{lem:F_t_in_POLY_are_good}, we note that all lemmas
and theorems subsequent to Lemma \ref{lem:rootQ} in Section
\ref{sec:BC_EPIR_fails_frequently_for_g} can be generalized  for any
polynomial $F(t)\in \mathcal{P}$. Therefore, we have that

\begin{theorem}\label{thm:BC_is_not_correct_for_gg}
The restricted version  does not satisfy the correctness requirement if $F(t)\in
\mathcal{P}$.
\end{theorem}

\subsection{Extension to Any Characteristic $p>2$}
We have stressed in Section \ref{sec:notations} that our methodology
is applicable when the characteristic of all related finite fields
is any prime $p$. For example, it is obvious that
we have an analog of Lemma \ref{lem:lower_bound_large_x} for any
characteristic $p>2$.
Let $\omega_p(n)$ be an analog of the
function $\omega(n)$ when the characteristic of all related finite
fields is a prime $p>2$.
Then the following theorem holds as
well.
\begin{theorem}\label{thm:lower_bound_eta_for_any_n_p}

We have that $\eta(n,g,g)\geq 1-\omega_p(n)$ for every integer
$n\geq 2$, where $g\in {\rm GF}(p^n)$ is primitive and $p$ is an
arbitrary prime number.
\end{theorem}
It follows that Theorem \ref{thm:BC_is_not_correct_for_gg} also
holds when the  characteristic of all related finite fields is any
prime  $p>2$.

\section{Conclusion}

In this paper, we  show that the Bringer-Chabanne EPIR protocol does not satisfy the correctness requirement.
 To simplify the argument, we give a restricted version of  the Bringer-Chabanne EPIR protocol   . If the original protocol satisfies the correctness requirement, then so does the restricted version.
We  show that the restricted version fails frequently   if the
polynomial to be evaluated has some special property. This allows us
to get the expected conclusion, i.e., the Bringer-Chabanne EPIR
protocol    does not satisfy the correctness requirement.

\end{document}